\documentclass[runningheads,a4paper]{llncs}

\usepackage[ruled,vlined,linesnumbered]{algorithm2e}
\usepackage{bbm,cite}
\usepackage{graphicx}

\usepackage{amssymb}
\usepackage{graphicx}

\begin{document}

\mainmatter  

\title{Computing Unique Maximum Matchings in $O(m)$ time for K\"{o}nig-Egerv\'{a}ry
Graphs and Unicyclic Graphs}

\titlerunning{Unique Maximum Matchings}

\authorrunning{Vadim E. Levit and Eugen Mandrescu}

\author{Vadim E. Levit\inst{1} \and Eugen Mandrescu\inst{2}}

\institute{Department of Computer Science and Mathematics \\ Ariel University, ISRAEL \\
\email{levitv@ariel.ac.il} \and Department of Computer Science \\ Holon Institute of Technology,
ISRAEL \\ \email{eugen\_m@hit.ac.il}}

\maketitle

\begin{abstract}
Let $\alpha\left(  G\right)  $ denote the maximum size of an independent set
of vertices and $\mu\left(  G\right)  $ be the cardinality of a maximum
matching in a graph $G$. A matching saturating all the vertices is a
\textit{perfect matching}. If $\alpha\left(  G\right)  +\mu\left(  G\right)
=\left\vert V(G)\right\vert $, then $G$ is called a
K\"{o}nig-Egerv\'{a}ry graph. A graph is \textit{unicyclic} if it has a unique cycle. \\

It is known that a maximum matching can be found in $O(m\bullet\sqrt{n})$ time for a
graph with $n$ vertices and $m$ edges. Bartha \cite{Bartha2010} conjectured that a unique perfect matching, if it
exists, can be found in $O(m)$ time. \\

In this paper we validate this conjecture for K\"{o}nig-Egerv\'{a}ry graphs
and unicylic graphs. We propose a variation of Karp-Sipser leaf-removal
algorithm \cite{KarpSpiser1981}, which ends with an empty graph if and only if
the original graph is a K\"{o}nig-Egerv\'{a}ry graph with a unique perfect
matching (obtained as an output as well). \\

We also show that a
unicyclic non-bipartite graph $G$ may have at most one perfect matching, and
this is the case where $G$ is a K\"{o}nig-Egerv\'{a}ry graph.
\end{abstract}

\section{Introduction}

Throughout this paper $G$ is a simple (i.e., finite, undirected, loopless and
without multiple edges) graph with vertex set $V(G)$ and edge set $E(G)$. If
$X\subseteq V$, then $G[X]$ is the subgraph of $G$ induced by $X$. If $A,B$
$\subseteq V\left(  G\right)  $ and $A\cap B=\emptyset$, then $(A,B)$ stands
for the set
\[
\{e=ab:a\in A,b\in B,e\in E\left(  G\right)  \}.
\]
The \textit{neighborhood} $N(v)$\ of a vertex $v\in V\left(  G\right)  $ is
the set $\{u:u\in V$\ and $vu\in E\}$. For $A\subseteq V\left(  G\right)  $,
we denote
\[
N_{G}(A)=\{v\in V\left(  G\right)  -A:N(v)\cap A\neq\emptyset\}
\]
and $N_{G}[A]=A\cup N(A)$, or for short, $N(A)$ and $N[A]$. If $N(v)=\{u\}$,
then $v$ is a \textit{leaf} and $uv$ is a \textit{pendant edge} of $G$. Let
$\mathrm{leaf}(G)$\ stand for the set of all leaves in $G$.
A graph is \textit{unicyclic} if it has a unique cycle.
Unicyclic graphs keep enjoying plenty of interest, as one can see, for instance,
in \cite{Belar2010,LevMan2012a,LevMan2014,Li2010,Wu2010,Zhai2010}.

An \textit{independent} set in $G$ is a set of pairwise non-adjacent vertices.
An independent set of maximum size is a \textit{maximum independent set} of
$G$, and $\alpha(G)$ is the cardinality of a maximum independent set in $G$.
Let $\Omega(G)$ stand for the set of all maximum independent sets of $G$, and
\textrm{core}$(G)=%
{\displaystyle\bigcap}
\{S:S\in\Omega(G)\}$ \cite{levm3}.

A \textit{matching} in a graph $G$ is a set $M\subseteq E\left(  G\right)  $
such that no two edges of $M$ share a common vertex. A \textit{maximum
matching} is a matching of maximum cardinality. By $\mu(G)$ is denoted the
size of a maximum matching. A matching is \textit{perfect} if it saturates all
the vertices of the graph.

$G$ is a \emph{K\"{o}nig-Egerv\'{a}ry graph} provided
$\alpha(G)+\mu(G)=\left\vert V(G)\right\vert $ \cite{Dem,Ster}. As a
well-known example, every bipartite graph is a K\"{o}nig-Egerv\'{a}ry graph
\cite{Eger,Koen}. Several properties of K\"{o}nig-Egerv\'{a}ry graphs are presented in
\cite{Korach2006,Larson2011,LevMan2012b,LevMan2013,Lov1983}.

\begin{theorem}
\cite{LevMan2001}\label{th2} A connected bipartite graph $G$ has a perfect
matching if and only if \textrm{core}$(G)=\emptyset$.
\end{theorem}

Theorem \ref{th2} may fail for non-bipartite K\"{o}nig-Egerv\'{a}ry graphs;
e.g., the graphs $G_{1}$ and $G_{2}$ from Figure \ref{fig12} have
\textrm{core}$(G_{1})=\{a\}$, and \textrm{core}$(G_{2})=\{u\}$.

\begin{figure}[h]
\setlength{\unitlength}{1.0cm}
\begin{picture}(5,1.5)\thicklines
\multiput(1.5,0)(1,0){4}{\circle*{0.29}}
\multiput(2.5,1)(1,0){2}{\circle*{0.29}}
\put(1.5,0){\line(1,0){3}}
\put(3.5,1){\line(1,-1){1}}
\put(2.5,1){\line(1,0){1}}
\put(2.5,0){\line(0,1){1}}
\put(1.2,0){\makebox(0,0){$a$}}
\put(3,1.6){\makebox(0,0){$G_{1}$}}
\multiput(6,0)(1,0){6}{\circle*{0.29}}
\multiput(7,1)(1,0){2}{\circle*{0.29}}
\multiput(10,1)(1,0){2}{\circle*{0.29}}
\put(6,0){\line(1,0){5}}
\put(7,0){\line(0,1){1}}
\put(7,1){\line(1,0){1}}
\put(8,1){\line(1,-1){1}}
\put(10,0){\line(0,1){1}}
\put(11,0){\line(0,1){1}}
\put(10,1){\line(1,0){1}}
\put(5.7,0){\makebox(0,0){$u$}}
\put(8.5,1.6){\makebox(0,0){$G_{2}$}}
\end{picture}\caption{Both $G_{1}$ and $G_{2}$ are K\"{o}nig-Egerv\'{a}ry
graphs with perfect matchings.}
\label{fig12}
\end{figure}
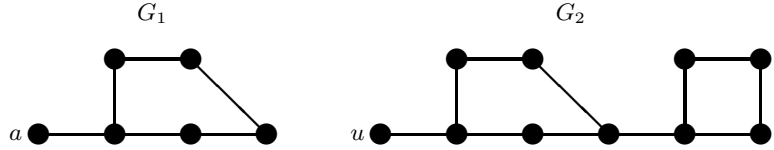

In a K\"{o}nig-Egerv\'{a}ry graph, maximum matchings have a special property,
emphasized by the following statement.

\begin{lemma}
\cite{LevMan2003}\label{match} Every maximum matching $M$ of a
K\"{o}nig-Egerv\'{a}ry graph $G$ is contained in each $(S,V\left(  G\right)
-\nolinebreak S)$ and $\left\vert M\right\vert =\left\vert V\left(  G\right)
-S\right\vert $, where $S\in\Omega(G)$.
\end{lemma}

If for every two incident edges of a cycle $C$ exactly one of them belongs to
a matching $M$, then $C$ is called an $M$\textit{-alternating cycle
}\cite{Krogdahl}. It is clear that an $M$-alternating cycle should be of even
length. A matching $M$ in $G$ is called \textit{alternating cycle-free} if $G$
has no $M$-alternating cycle. For example, the matching $\{ab,cd,ef\}$ of the
graph $G$ from Figure \ref{fig111} is alternating cycle-free.

\begin{figure}[h]
\setlength{\unitlength}{1.0cm}
\begin{picture}(5,1.5)\thicklines
\multiput(2,0)(1,0){4}{\circle*{0.29}}
\multiput(3,1)(1,0){3}{\circle*{0.29}}
\put(2,0){\line(1,0){3}}
\put(4,0){\line(1,1){1}}
\put(3,1){\line(1,0){1}}
\multiput(3,0)(1,0){2}{\line(0,1){1}}
\put(1.7,0){\makebox(0,0){$a$}}
\put(2.7,0.3){\makebox(0,0){$b$}}
\put(3.7,0.3){\makebox(0,0){$c$}}
\put(5.3,0){\makebox(0,0){$d$}}
\put(2.7,1){\makebox(0,0){$e$}}
\put(4.3,1){\makebox(0,0){$f$}}
\put(5.3,1){\makebox(0,0){$g$}}
\put(1.1,0.5){\makebox(0,0){$G$}}
\multiput(7.5,0)(1,0){3}{\circle*{0.29}}
\multiput(8.5,1)(1,0){3}{\circle*{0.29}}
\put(7.5,0){\line(1,0){2}}
\put(8.5,1){\line(1,0){2}}
\multiput(8.5,0)(1,0){2}{\line(0,1){1}}
\put(7.2,0){\makebox(0,0){$u$}}
\put(8.2,0.3){\makebox(0,0){$v$}}
\put(9.8,0){\makebox(0,0){$t$}}
\put(10.85,1){\makebox(0,0){$w$}}
\put(8.2,1){\makebox(0,0){$y$}}
\put(9.8,0.7){\makebox(0,0){$x$}}
\put(6.7,0.5){\makebox(0,0){$H$}}
\end{picture}\caption{The unique cycle of $H$ is alternating with respect to
the matching $\{yv,tx\}$.}
\label{fig111}
\end{figure}
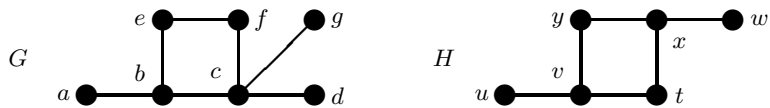

A matching
\[
M=\{a_{i}b_{i}:a_{i},b_{i}\in V(G),1\leq i\leq k\}
\]
of graph $G$ is called \textit{a uniquely restricted matching} if $M$ is the
unique perfect matching of $G[\{a_{i},b_{i}:1\leq i\leq k\}]$ \cite{GolHiLew}.
For bipartite graphs, this notion was first introduced\emph{\ }in \cite{Krogdahl}, under the
name \textit{clean matching}. It appears also in the context of matrix theory,
as a \textit{constrained matching} \cite{HerSchn}.

\begin{theorem}
\cite{GolHiLew}\label{th9} A matching is uniquely restricted if and only if it
is \textit{alternating} cycle-free.
\end{theorem}

For instance, all the maximum matchings of the graph $G$ in Figure
\ref{fig111} are uniquely restricted, while the graph $H$ from the same figure
has both uniquely restricted maximum matchings (e.g., $\{uv,xw\}$) and
non-uniquely restricted maximum matchings (e.g., $\{xy,tv\}$).

\begin{lemma}
\cite{Chechlarova1991} If a graph without isolated vertices has a unique
maximum matching, then this matching is perfect.
\end{lemma}

To find a maximum matching one needs $O(m\bullet\sqrt{n})$ time for a
graph with $n$ vertices and $m$ edges \cite{MiVa80}. If our goal is to check
whether a graph possesses a unique perfect matching, then we can do\ better. A
most efficient unique perfect matching algorithm runs in $O(m\bullet
log^{4}n)$ time \cite{GaKaTarjan2001}. An $O(m)$ algorithm is given for the
special cases of chestnut and elementary soliton graphs in
\cite{BarthaKresz2009}. It is known that bipartite graphs with a unique
maximum matching can be recognized by an $O(m)$ algorithm as well
\cite{Chechlarova1991}.

\begin{conjecture}
\label{conj}\cite{Bartha2010} For a graph of size $m$, a unique perfect
matching, if it exists, can always be found in $O(m)$ time.
\end{conjecture}

In what follows, we validate Conjecture \ref{conj} for both
K\"{o}nig-Egerv\'{a}ry graphs and unicyclic graphs.

\section{Results}

According to Theorem \ref{th9}, if $M$ is a perfect matching in graph $G$,
then $M$ is unique if and only if no cycle of $G$ is alternating with respect
to $M$. Therefore, a perfect matching in a tree, if any, must be unique.

\begin{lemma}
\label{lem1}\cite{Chechlarova1991,LevMan45} If $G=(A,B,E)$ is a bipartite
graph having a unique perfect matching, then $A\cap\mathrm{leaf}%
(G)\neq\emptyset$ and $B\cap\mathrm{leaf}(G)\neq\emptyset$.
\end{lemma}

In other words, a bipartite graph with a unique perfect matching must have at
least two leaves. Notice that there exist non-bipartite graphs with unique
perfect matchings and without leaves. For an example, see the graph $G_{2}$
from Figure \ref{fig10}.

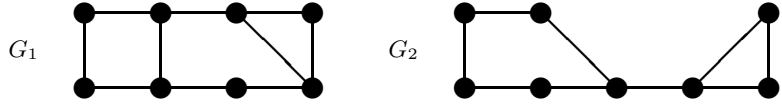
\begin{figure}[h]
\setlength{\unitlength}{1.0cm}
\begin{picture}(5,1.5)\thicklines
\multiput(2,0)(1,0){4}{\circle*{0.29}}
\multiput(2,1)(1,0){4}{\circle*{0.29}}
\put(2,0){\line(1,0){3}}
\put(2,0){\line(0,1){1}}
\put(4,1){\line(1,-1){1}}
\put(2,1){\line(1,0){3}}
\put(3,0){\line(0,1){1}}
\put(5,0){\line(0,1){1}}
\put(1.2,0.5){\makebox(0,0){$G_{1}$}}
\multiput(7,0)(1,0){5}{\circle*{0.29}}
\multiput(7,1)(1,0){2}{\circle*{0.29}}
\put(11,1){\circle*{0.29}}
\put(7,0){\line(1,0){4}}
\put(7,0){\line(0,1){1}}
\put(7,1){\line(1,0){1}}
\put(8,1){\line(1,-1){1}}
\put(10,0){\line(1,1){1}}
\put(11,0){\line(0,1){1}}
\put(6.2,0.5){\makebox(0,0){$G_{2}$}}
\end{picture}\caption{Both $G_{1}$ and $G_{2}$ have perfect matchings.}
\label{fig10}
\end{figure}

The following lemma, firstly presented in \cite{LevMan07}, shows that every
K\"{o}nig-Egerv\'{a}ry graph with a unique perfect matching has at least one
leaf (see, for example, the graph $G_{1}$, depicted in Figure \ref{fig12}). We
give a proof here for the sake of self-containment.

\begin{lemma}
\cite{LevMan07} \label{lem11}If $G$ is a K\"{o}nig-Egerv\'{a}ry graph with a
unique perfect matching, then $S\cap\mathrm{leaf}(G)\neq\emptyset$ holds for
every $S\in\Omega(G)$.
\end{lemma}

\begin{proof}
Let
\[
M=\{a_{i}b_{i}:1\leq i\leq\mu(G)\}
\]
be the unique perfect matching of $G$
and $S\in\Omega(G)$. Since $G$ is a K\"{o}nig-Egerv\'{a}ry graph, it follows
that
\[
\left\vert M\right\vert =\mu(G)=\alpha(G)=\left\vert S\right\vert .
\]
By Lemma \ref{match}, $M\subseteq(S,V\left(  G\right)  -S)$ and, therefore, we
may assume that
\[
S=\{a_{i}:1\leq i\leq\mu(G)\}.
\]
Suppose that $S\cap\mathrm{leaf}(G)=\emptyset$. Hence, $\left\vert
N(a_{i})\right\vert \geq2$ for every $a_{i}\in S$. Under these conditions, we
shall build an $M$-alternating cycle $C$. We begin with the edge $a_{1}b_{1}$;
since $\left\vert N(a_{1})\right\vert \geq2$, there is some $b\in
(V-S-\{b_{1}\})\cap N(a_{1})$, say $b_{2}$. We continue with $a_{2}b_{2}\in
M$. Further, $N(a_{2})$ contains some $b\in(V-S-\{b_{2}\})$. If $b_{1}\in
N(a_{2})$, we are done, because $G[\{a_{1},a_{2},b_{1},b_{2}\}]=C_{4}$.
Otherwise, we may suppose that $b=b_{3}$, and we add to the growing cycle the
edge $a_{3}b_{3}$. Since $G$ has a finite number of vertices, after a number
of edges from $M$, we must find some edge $a_{k}b_{j}$ having $1\leq j<k$. So,
the cycle $C$ we found has
\[
V(C)=\{a_{i},b_{i}:j\leq i\leq k\},\ \\
\]
\[
E(C)=\{a_{i}b_{i}:j\leq i\leq k\}\cup\{a_{i}b_{i+1}:j\leq i<k\}\cup
\{a_{k}b_{j}\}.
\]
Clearly, $C$ is an $M$-alternating cycle. Hence, by Theorem \ref{th9}, $M$ is
not unique, which contradicts the hypothesis on $M$.
\end{proof}

It is worth mentioning that Lemma \ref{lem11} may fail for K\"{o}nig-Egerv\'{a}ry
graphs having more than one perfect matching;
e.g., the graph $G_{1}$ from Figure \ref{fig10}.

Lemma \ref{lem11} plays a key-role in the following procedure checking whether
a K\"{o}nig-Egerv\'{a}ry graph has a unique perfect matching. It reads as
follows: as long as there a leaf $w$, add the edge connecting $w$ with its
only neighbor to a matching, and remove both vertices from the graph. If we
end up with the empty graph, then we have found a unique perfect matching, and
validated that our input is a K\"{o}nig-Egerv\'{a}ry graph. Otherwise, either
the graph is a non-K\"{o}nig-Egerv\'{a}ry graph, or it has more than one
maximum matching. Actually, this procedure is a variation of the Karp-Sipser
algorithm \cite{KarpSpiser1981}.

\begin{algorithm}[h!]
\label{alg:UPM}
\caption{Unique Perfect Matching}

\KwIn{A graph $G$;}
\KwOut{A unique perfect matching $M$ of $G$, and an evidence that
$G$ is a K\"{o}nig-Egerv\'{a}ry graph;

\qquad\qquad \ \textbf{otherwise}, a non-empty subgraph of $G$ without leaves.}

\BlankLine

{
    \nllabel{line1}Initialize a one-dimensional boolean array $Vertex[\ ]$ with
$Vertex\left[  i\right]  =True$ for $1\leq i\leq n$. It will be updated further as the set of vertices $V(G)$ changes. \\

    \nllabel{line1.1}Find the set $\mathrm{leaf}(G)$\ and present it like a \textbf{Queue}. \\

    \nllabel{line1.2}$M \leftarrow \emptyset$

    \While{$\mathrm{leaf}(G)\neq\emptyset$}
    {\nllabel{line5}

        \nllabel{line6}Take the first vertex from $\mathrm{leaf}(G)$, say $v$.\\

            \If{$v\in V\left(  G\right)  $ (or, in other words, \textbf{if} $Vertex\left[  v\right]  =True$)}
            {
                $V\left(  G\right)  \leftarrow V\left(  G\right) -N\left[  v\right]$ \\
                $M \leftarrow M\cup\left(  v,N\left(  v\right)  \right)$ \\
                $\mathrm{leaf}(G) \leftarrow \mathrm{leaf}(G)-v$ \\
                Add all new leaves of $G$ from $N\left(N\left(  v\right)  \right)  -v$ to $\mathrm{leaf}(G)$.

            }
        }

        \If{$V\left(  G\right)  =\emptyset$}
            {
                $M$ is a unique perfect matching and $G$ is a K\"{o}nig-Egerv\'{a}ry graph.

            }
        \ElseIf{$G$ is a K\"{o}nig-Egerv\'{a}ry graph}
                {
                        The number of maximum matchings is greater than $1$. \\
                        \Else
                              {
                              Nothing specific can be said on the number of maximum matchings.
                              }
                }

}
\end{algorithm}

\begin{figure}[h]
\setlength{\unitlength}{1cm}
\begin{picture}(5,1.2)\thicklines
\multiput(0.5,0)(1,0){4}{\circle*{0.29}}
\multiput(0.5,1)(1,0){4}{\circle*{0.29}}
\put(0.5,0){\line(1,0){1}}
\put(0.5,0){\line(0,1){1}}
\put(0.5,0){\line(1,1){1}}
\put(1.5,0){\line(1,1){1}}
\put(1.5,0){\line(0,1){1}}
\put(2.5,0){\line(1,0){1}}
\put(2.5,0){\line(0,1){1}}
\put(2.5,0){\line(1,1){1}}
\put(3.5,0){\line(0,1){1}}
\put(2,1.5){\makebox(0,0){$G_{1}$}}
\multiput(4.8,0)(1,0){3}{\circle*{0.29}}
\multiput(4.8,1)(1,0){3}{\circle*{0.29}}
\put(4.8,0){\line(1,0){2}}
\put(4.8,0){\line(0,1){1}}
\put(5.8,0){\line(0,1){1}}
\put(5.8,1){\line(1,0){1}}
\put(5.8,1){\line(1,-1){1}}
\put(6.8,0){\line(0,1){1}}
\put(5.8,1.5){\makebox(0,0){$G_{2}$}}
\multiput(8,0)(1,0){5}{\circle*{0.29}}
\multiput(8,1)(1,0){5}{\circle*{0.29}}
\put(8,0){\line(1,0){1}}
\put(8,0){\line(0,1){1}}
\put(8,0){\line(1,1){1}}
\put(9,0){\line(0,1){1}}
\put(9,0){\line(1,1){1}}
\put(10,0){\line(0,1){1}}
\put(10,0){\line(1,0){2}}
\put(11,1){\line(1,0){1}}
\put(11,1){\line(1,-1){1}}
\put(11,0){\line(0,1){1}}
\put(11,0){\line(1,1){1}}
\put(12,0){\line(0,1){1}}
\put(10,1.5){\makebox(0,0){$G_{3}$}}
\end{picture}\caption{Both $G_{1}$ and $G_{2}$ are K\"{o}nig-Egerv\'{a}ry graph, but
only $G_{1}$ has a unique perfect matching.}
\label{fig123}
\end{figure}
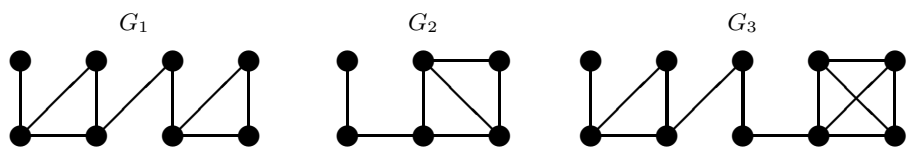

When Algorithm \ref{alg:UPM} is applied to the graphs from Figure \ref{fig123},
only for $G_{1}$ it ends with $V\left(  G\right)  =\emptyset$. On the other hand,
$G_{3}$ has perfect matchings, but it is not a K\"{o}nig-Egerv\'{a}ry graph.

Notice that there exist non-K\"{o}nig-Egerv\'{a}ry graphs having a
unique perfect matching, with or without leaves (for instance, the graphs
$G_{2},G_{3}$ in Figure \ref{fig1122}).

\begin{figure}[h]
\setlength{\unitlength}{1.0cm}
\begin{picture}(5,1.8)\thicklines
\multiput(0.5,0.5)(1,0){3}{\circle*{0.29}}
\put(2.5,1.5){\circle*{0.29}}
\put(0.5,0.5){\line(1,0){2}}
\put(1.5,0.5){\line(1,1){1}}
\put(2.5,0.5){\line(0,1){1}}
\put(1.5,0){\makebox(0,0){$G_{1}$}}
\multiput(3.5,0.5)(1,0){4}{\circle*{0.29}}
\multiput(4.5,1.5)(1,0){2}{\circle*{0.29}}
\put(3.5,0.5){\line(1,0){3}}
\put(3.5,0.5){\line(1,1){1}}
\put(4.5,0.5){\line(0,1){1}}
\put(5.5,0.5){\line(0,1){1}}
\put(5.5,1.5){\line(1,-1){1}}
\put(5,0){\makebox(0,0){$G_{2}$}}
\multiput(7.5,0.5)(1,0){5}{\circle*{0.29}}
\multiput(8.5,1.5)(1,0){2}{\circle*{0.29}}
\put(11.5,1.5){\circle*{0.29}}
\put(7.5,0.5){\line(1,0){4}}
\put(7.5,0.5){\line(1,1){1}}
\put(8.5,0.5){\line(0,1){1}}
\put(9.5,0.5){\line(0,1){1}}
\put(10.5,0.5){\line(-1,1){1}}
\put(11.5,0.5){\line(0,1){1}}
\put(9.5,0){\makebox(0,0){$G_{3}$}}
\end{picture}\caption{Each of the graphs $G_{1},G_{2},G_{3}$ {has a unique
perfect matching, but only }$G_{1}$ is a K\"{o}nig-Egerv\'{a}ry\ graph.}
\label{fig1122}
\end{figure}
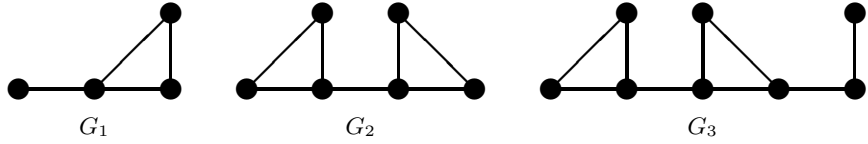

\begin{theorem}
\label{th1}Algorithm \ref{alg:UPM} ends with $V\left(  G\right)  =\emptyset$
if and only if $G$\ is a K\"{o}nig-Egerv\'{a}ry graph with a unique perfect matching.
\end{theorem}

\begin{proof}
\textit{If part}. Let
\[
M=\{a_{i}b_{i}:1\leq i\leq\mu(G)\}
\]
be the unique
perfect matching of the K\"{o}nig-Egerv\'{a}ry graph $G$, and $S\in\Omega(G)$.
According to Lemma \ref{match},
\[
M\subseteq(S,V\left(  G\right)-S),
\]
and by Lemma \ref{lem11}, we may assume that
\[
a_{1}\in S\cap\mathrm{leaf}(G).
\]
Clearly,
\[
G-\{a_{1},b_{1}\}=G-N[a_{1}]
\]
is still a
K\"{o}nig-Egerv\'{a}ry graph with the unique perfect matching, $M-\{a_{1}b_{1}\}$.
Lemma \ref{lem11} assures that the graph $G-N[a_{1}]$ has one leaf
(at least), say $a_{2}$, and
\[
G-N[a_{1}]-N[a_{2}]
\]
is again a
K\"{o}nig-Egerv\'{a}ry graph having a unique perfect matching, namely,
\[
M-\{a_{1}b_{1},a_{2}b_{2}\}.
\]
Hence, repeating this procedure $\mu(G)$ times,
we finally arrive at $G=\emptyset$, i.e., Algorithm \ref{alg:UPM} correctly
finds the unique perfect matching of $G$.

\textit{Only if part}. We proceed by induction on $m=\left\vert E\left(
G\right)  \right\vert $.

The result is true for $m=1$.

Assume that the assertion holds for every graph on $m\geq1$ edges, and let
$G$\ be a graph on $m+1$ edges, for which Algorithm \ref{alg:UPM} ends with
$V\left(  G\right)  =\emptyset$.

If $a\in\mathrm{leaf}(G)$ and $ab\in E\left(
G\right)  $ is the first pendant edge that Algorithm \ref{alg:UPM} deletes
from $G$, then the remaining graph
\[
G-ab=\left(  V\left(  G\right)  -\left\{  a,b\right\}  ,E\left(  G\right)
-\left\{  ab\right\}  \right)  =\left(  W,U\right)
\]
has $m$ edges and Algorithm \ref{alg:UPM} ends with $W=\emptyset$.
Consequently, by the induction hypothesis, $G-ab$ is a K\"{o}nig-Egerv\'{a}ry
graph with a unique perfect matching, say $M$. Hence, there is no
$M$-alternating cycle in $G-ab$.

Clearly, $M\cup\left\{  ab\right\}$ is a
perfect matching in $G$, and it is unique, because the vertex $a\in V\left(
G\right)  $ can be saturated only by the pendant edge $ab$, and there do not exist
$M\cup\left\{  ab\right\}  $-alternating cycles in $G$. In addition, if
$S$ is a maximum independent set in $G-ab$, then $S\cup\left\{  a\right\}  $
is a maximum independent set in $G$, and thus
\[
\left\vert V\left(  G\right)  -\left\{  a,b\right\}  \right\vert
+2=\alpha\left(  G-ab\right)  +\mu\left(  G-ab\right)  +2=
\]
\[
=\left\vert S\cup\left\{  a\right\}  \right\vert +\left\vert M\cup\left\{
ab\right\}  \right\vert \leq\alpha\left(  G\right)  +\mu\left(  G\right)
\leq\left\vert V\left(  G\right)  \right\vert ,
\]
i.e., $G$ is a K\"{o}nig-Egerv\'{a}ry graph.
\end{proof}

\begin{corollary}
\label{cor}A K\"{o}nig-Egerv\'{a}ry graph $G$ has a unique perfect matching if
and only if Algorithm \ref{alg:UPM} ends with $V\left(  G\right)
=\emptyset$.
\end{corollary}

Graphs $G_{2}$ from Figure \ref{fig1122} and $G_{3}$ from Figure \ref{fig123}
show that Algorithm \ref{alg:UPM} can not estimate the number of maximum
matchings of non-K\"{o}nig-Egerv\'{a}ry graphs.

It is known that if a graph does not possess a perfect matching, then it has two
maximum matchings at least. Consequently, by Corollary \ref{cor}, if $G$ is a
K\"{o}nig-Egerv\'{a}ry graph and Algorithm \ref{alg:UPM} returns
$V\neq\emptyset$, then its number of maximum matchings is greater than $1$.

Since every edge of the graph $G$\ is in use no more than twice in Algorithm
\ref{alg:UPM}, we obtain the following.

\begin{theorem}
Given a K\"{o}nig-Egerv\'{a}ry graph with $m$ edges, Algorithm \ref{alg:UPM}
decides whether it has a unique perfect matching in $O(m)$ time.
\end{theorem}

Clearly, if the cycle of a unicyclic graph $G$ is even, then $G$ is bipartite, and hence it is a
K\"{o}nig-Egerv\'{a}ry graph.

\begin{proposition}
If a unicyclic non-bipartite graph $G$ has $M$ as a perfect matching, then $M$
is unique and $G$ is a K\"{o}nig-Egerv\'{a}ry graph.
\end{proposition}

\begin{proof}
Let $C$ be the unique cycle of $G$. Since $C$ is an odd cycle, Theorem
\ref{th9} ensures that $M$ is unique.

Notice that in order to show that $G$ is a K\"{o}nig-Egerv\'{a}ry graph it is
enough to prove that $G$ has an independent set of size equal to $\left\vert
M\right\vert $.

Let
\[
M_{C}=\left\{  uv\in M:\left\{  u,v\right\}  \cap V\left(  C\right)
\neq\emptyset\right\},
\]
and $H$ be the subgraph of $G$ induced by the
vertices saturated by $M_{C}$.

Since $C$ is odd and $M$ covers all the
vertices of $G$, it follows that $M_{C}$ is a (unique) perfect matching in $H$.
In addition, $H$ has one leaf at least, say $a$. Removing the vertex $a$
together with its neighbor may be considered as the first step of Algorithm
\ref{alg:UPM}. Clearly, $H-a$ is a forest with a perfect matching.
Consequently, by Theorem \ref{th1} Algorithm \ref{alg:UPM} terminates with
empty graph, which, in turn, means that $H$ is a K\"{o}nig-Egerv\'{a}ry graph.

Each connected component of the graph $G-H$ is a tree $T$ with a perfect
matching. Consequently, if $uv\in E$ is such that $u\in V\left(  H\right)  $
and $v\in V\left(  T\right)  $, then, by Theorem \ref{th2}, claiming in our
case that \textrm{core}$(T)=\emptyset$, there must be a maximum independent
set $S_{T}$ in $T$ with $v\notin S_{T}$. If $\Gamma$ denotes the family of all
connected components of the graph $G-H$, we obtain that
\[
A=S_{H}\cup\left(
{\displaystyle\bigcup}
\left\{  S_{T}:T\in\Gamma\right\}  \right)
\]
is an independent set of $G$,
and $\left\vert A\right\vert =\left\vert M\right\vert $. Therefore, $G$ is a
K\"{o}nig-Egerv\'{a}ry graph.
\end{proof}

\begin{corollary}
If $G$ is unicyclic and has $m$ edges, then Algorithm \ref{alg:UPM} decides
whether it has a unique perfect matching in $O(m)$ time.
\end{corollary}

\section{Conclusions}

In this paper we have validated Conjecture \ref{conj} claiming that a unique
perfect matching, if it exists, can always be found in $O(m)$ time, for both
K\"{o}nig-Egerv\'{a}ry graphs and unicyclic graphs.

Very well-covered graphs, a subclass of K\"{o}nig-Egerv\'{a}ry graphs with
perfect matchings, can be recognized in polynomial time. Namely, to recognize
a graph as being very well-covered, we just need to show that it has a perfect
matching $M$ such that for every edge $xy\in M$: $N(x)\cap N(y)=\emptyset$,
and each $v\in N(x)-\{y\}$ is adjacent to all vertices of $N(y)-\left\{
x\right\}  $ \cite{Favaron}. To check this property one has to handle
$O\left(  n^{3}\right)  $ pairs of vertices in the worst case. Recently, very
well-covered graphs with unique perfect matching proved their importance in
\cite{LevMan2011,LevMan2012}. It is an open problem to recognize a
very well-covered graph with a unique perfect matching, faster than in
$O\left(  n^{3}\right)$ time, when the input is a general graph.

\end{document}